\documentclass{ifacconf}
\usepackage[utf8x]{inputenc} 
\usepackage{graphicx}      
\usepackage{natbib}        
 \usepackage{amsfonts}
\usepackage{amsmath}
 \usepackage{amssymb}
 \usepackage{enumerate}
 \usepackage{flafter}
 \usepackage{fleqn}
 \usepackage{latexsym}
 \usepackage{mathrsfs}
 \usepackage{psfrag}
 \usepackage{subfigure}
 \usepackage{url}
  \usepackage{color}
  \usepackage{times}
 \usepackage{flushend}
 \usepackage{bm}
 \usepackage{multirow}
 \usepackage{multicol}
 \usepackage{arydshln} 
 \usepackage{soul}

\newtheorem{mydef}{Definition}

\newtheorem{mythm}{Theorem}

\begin{document}
\begin{frontmatter}

\title{Precise Motion Control of Wafer Stages via Adaptive Neural Network and Fractional-Order Super-Twisting Algorithm} 

\thanks[footnoteinfo]{This work is partially supported by CSC (China Scholarship Council).}

\author[First,Second]{Zhian Kuang} 
\author[Second]{Liting Sun}
\author[First]{Huijun Gao}
\author[Second]{Masayoshi Tomizuka}

\address[First]{Research Institute of Intelligent Control and Systems, Harbin Institute of Technology, Harbin 150001, P.R. China (e-mail: zhiankuang@foxmail.com)}
\address[Second]{Mechanical Control System Lab, Mechanical Engineering Department, University of California, Berkeley, CA 94720, USA (e-mail:tomizuka@berkeley.edu)}


\begin{abstract}                
To obtain precise motion control of wafer stages, an adaptive neural network and fractional-order super-twisting control strategy is proposed. Based on sliding mode control (SMC), the proposed controller aims to address two challenges in SMC: 1) reducing the chattering phenomenon, and 2) attenuating the influence of model uncertainties and disturbances. For the first challenge, a fractional-order terminal sliding mode surface and a super-twisting algorithm are integrated into the SMC design. To attenuate uncertainties and disturbances, an add-on control structure based on the radial basis function (RBF) neural network is introduced. Stability analysis of the closed-loop control system is provided. Finally, experiments on a wafer stage testbed system are conducted, which proves that the proposed controller can robustly improve the tracking performance in the presence of uncertainties and disturbances compared to conventional and previous controllers.
\end{abstract}


\end{frontmatter}

\section{Introduction}
Photolithography is one of the most important processes for semiconductor manufacturing ~\citep{mishra2009projephotoction}. An exemplar photolithography system is shown in Fig.~\ref{fig:wafer_scanner}. A laser beam that goes through the integrated circuit patterns on the reticle is projected on the wafer so that the patterns are printed onto the wafer. During this process, the stage carrying the wafer (the wafer stage) needs to move steadily and precisely so that the patterns are printed accurately~\citep{oomen2013connecting}. 
\begin{figure}[http]
  \centering
   \includegraphics[width=200pt]{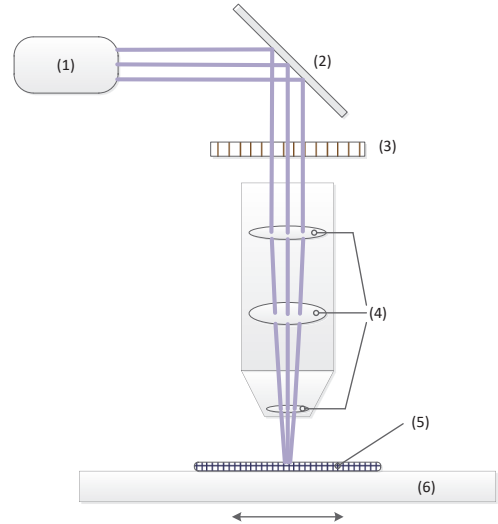}
   \caption{A schematic diagram of a photolithography system, where (1)~is a laser generator, (2)~a reflective mirror, (3)~a reticle, (4)~projection lenses, (5)~a wafer and (6)~a wafer stage.}
   \label{fig:wafer_scanner}
   \end{figure}
   
With the technological development of the semiconductor industry, manufacturers demand more precise performance from the wafer stage. To achieve the goal, researchers have developed many control strategies and applied them on the wafer stage, including iterative learning control (ILC)~\citep{heertjes2007robustness, mishra2009projection, sun2014selective, zheng2017design}, sliding mode control (SMC)~\citep{heertjes2014self, li2016state, wu2011sliding, ito2014sliding}, $H \infty$ feedback control~\citep{van2002multivariable}, multi-rate control~\citep{sun2016multirate} and so on. Among them, SMC has attracted great attention for its simple implementation and robust performance in the presence of uncertainties and external disturbances~\citep{fukushima2014sliding,kuang2019precise}.
Beyond the basic SMC structure~\citep{edwards1998sliding}, many advanced SMC strategies such as the modified reaching laws~\citep{yu2005continuous}, boundary layer technique~\citep{chen2002state}, and super-twisting algorithm (STA)~\citep{moreno2012strict} have also been proposed. The STA, focusing on improving the dynamics of the sliding variables, has been considered as one of the most effective approaches for the well-known chattering phenomenon~\citep{sun2018practical}. It is also robust with respect to bounded uncertainties and disturbances~\citep{shtessel2012novel,kuang2018contouring}, and has been implemented successfully in practice~\citep{shtessel2012novel,sadeghi2018super}.

To further improve the performance of SMC, fractional-order calculus is introduced to improve the state dynamics in the sliding surface and is combined with the STA~\citep{kuang2018high,wang2019adaptive}. Although there are some theoretical research and successful precedents for the application of fractional-order super-twisting algorithm (FOSTA)~\citep{wang2019adaptive,caponetto2015identification}, the parametric uncertainties are not always taken into consideration, or their bounds are assumed to be small. When the amplitudes of the uncertainties or disturbances are rather large, the sliding variable in STA cannot converge to the predefined sliding surface. Instead, it only converges to an uncertainty region around the sliding surface, which inevitably brings positioning error to the system~\citep{kuang2018simplified}. Previous researches have tried to reduce the range of the uncertainty region to improve the precision~\citep{sun2018practical}, but the negative influence from large uncertainties remains. 
 


Aiming to improve the control performance (i.e., high precision and robust performance) in the presence of large model uncertainties and disturbances, a novel adaptive neural network and fractional-order super-twisting algorithm (ANN-FSA) is proposed in this paper. Firstly, we use the radial basis function (RBF) neural network to approximate the uncertainties and disturbances in the system, and a corresponding fractional-order super-twisting controller is designed to compensate for uncertainties and disturbances. The stability of the proposed control strategy is also analyzed. Moreover, to guarantee the global convergence of the closed-loop system, an adaptive law is designed. At last, we apply the proposed controller to a wafer stage testbed. Experimental results show that the controller performs well and is robust against disturbances.

The remainder of this paper is organized as follows: Section~\ref{sec: system_model} provides the model of the wafer stage. Section~\ref{sec: controller design} presents the proposed controller, and the stability analysis of the controller. Section~\ref{sec:experiments} displays the experimental setup and the experimental results with the proposed controller. Finally, Section~\ref{sec:conclusion} presents conclusions.

\section{Model of Wafer Stage}\label{sec: system_model}

\begin{figure}[http]
  \centering
   \includegraphics[width=200pt]{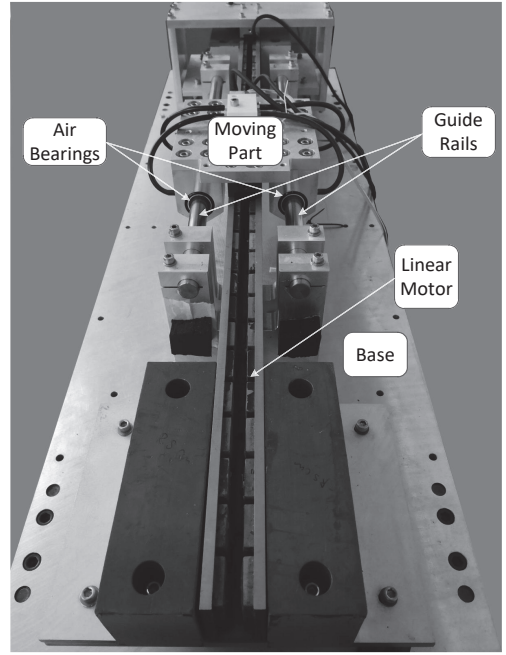}
   \caption{The experimental wafer stage testbed.}
   \label{fig:wafer_stage}
   \end{figure}
   
As Fig.~\ref{fig:wafer_stage} shows, the moving part of the wafer stage is driven by a permanent magnet linear synchronous motor (PMLSM).   
From~\cite{chen2019precision}, the mathematical model of a PMLSM is given as
\begin{align}
    F_e=M\ddot{p}+F_f+F_d,
\end{align}
where $F_e$ is the electromagnetic force generated by the linear motor, $M$ is the mass of the moving part, $p$ denotes the position of the moving part, $\ddot{p}$ is the second derivative of the position, i.e., the acceleration of the moving part, $F_f$ stands for the friction applied to the moving part, and $F_d$ is the external disturbance. The friction force $F_f$ is modelled as
\begin{align}
    F_f=F_c \textrm{sgn}(\dot{p})+K_v \dot{p},
\end{align}
where $F_c$ denotes the Coulomb friction force and $K_v$ is the viscous friction coefficient. As there is no direct contact between the moving stage and the guides, $F_c$ is very small and it is neglected in the following analysis in this paper. In modern linear motors, we can approximate the electromagnetic force $F_e$ to be proportional to the driving current with an additional nonlinear effect. Thus, $F_e$ is written as
\begin{align}
    F_e=K_e i-F_n(i),
\end{align}
where $K_e$ is electromagnetic coefficient and it is only related to the physical parameters of the PMLSM, $F_n(i)$ denotes the nonlinear thrust force including the force ripple.

From the analysis above, the mathematical model model becomes
\begin{align} \label{model:natural}
    \ddot{p}=& \frac{K_e}{M}i-\frac{K_v}{M}\dot{p}-\frac{F_d+F_n(i)}{M} \nonumber
    \\
    \triangleq&A\dot{p} +B u +d,
\end{align}
where $A=- \frac{K_v}{M}$, $B=\frac{K_e}{M}$, $d=-\frac{F_d+F_n(i)}{M}$ is the generalized disturbance, and $u=i$ is the control input.

In practical applications,  there are always errors between the identified results and the actual values, i.e., the parametric uncertainties. Here, theses uncertainties are assumed to be bounded:
\begin{align}
    A&=\bar{A}(1+\delta_{A}),\\
    B&=\bar{B}(1+\delta_B),
\end{align}
where $\bar{\cdot}$ stands for the identified nominal parameter, and $\delta_{\cdot}$ is an unknown constant.

The model~(\ref{model:natural}) is rewritten as
\begin{align} \label{model:final}
    \ddot{p}=\bar{A} \dot{p}+ \bar{B} u +\bar{A} \delta_A  \dot{p} + \bar{B}\delta_B u+d.
\end{align}

\section{Controller Design} \label{sec: controller design}

\subsection{The Proposed Controller}

For the convenience of further discussion, the definition of the fractional-order calculus is briefly reviewed.

\begin{mydef} (see \cite{podlubny1998fractional}) \label{def: fractional-order calculus}
  The definition of the $\xi$th order derivative for function $f(t)$ in Riemann-Liouville form is defined as
  \begin{equation}\begin{aligned}
\hspace*{-2em}D^{\xi}f(t)=\frac{1}{\gamma (m-\xi)}\frac{d^m}{dt^m}\int^t_{t_0}\frac{f(\tau)}{(t-\tau)^{\xi-m+1}}d\tau,
  \end{aligned}\end{equation}
  \begin{equation}\begin{aligned}
  \hspace*{-2em}  D^{-\xi}_{t}f(t)={}_{t_0}I^{\xi}_{t}f(t)=\frac{1}{\gamma (\xi)}\int^t_{t_0}\frac{f(\tau)}{(t-\tau)^{1-\xi}}d\tau,
  \end{aligned}\end{equation}
  where $\xi \in \mathbb{R}^+$ and $m-1<\xi<m$, $m\in \mathbb{N}$, and $\gamma (\bullet)$ is the Gamma function that $\gamma(\xi)=\int_0^{+\infty}t^{\xi-1}e^{-t}dt$. 
\end{mydef}


Fig.~\ref{fig:control structure} shows the proposed controller in this paper. It consists of three parts: a nominal-model-based equivalent controller, a super-twisting controller to reduce chattering, and a neural network controller to compensate for the uncertainties and disturbances. As shown in Fig.~\ref{fig:control structure}, the control law is given by:
\begin{align} \label{eq:overall controller}
u=u_{eq}+u_{nn}+u_{st},
\end{align}
where $u_{eq}$, $u_{st}$ and $u_{nn}$ are the control inputs from the equivalent controller, the super-twisting controller and the neural network controller, respectively. 
\begin{figure}[http]
	\centering
	\includegraphics[width=250pt]{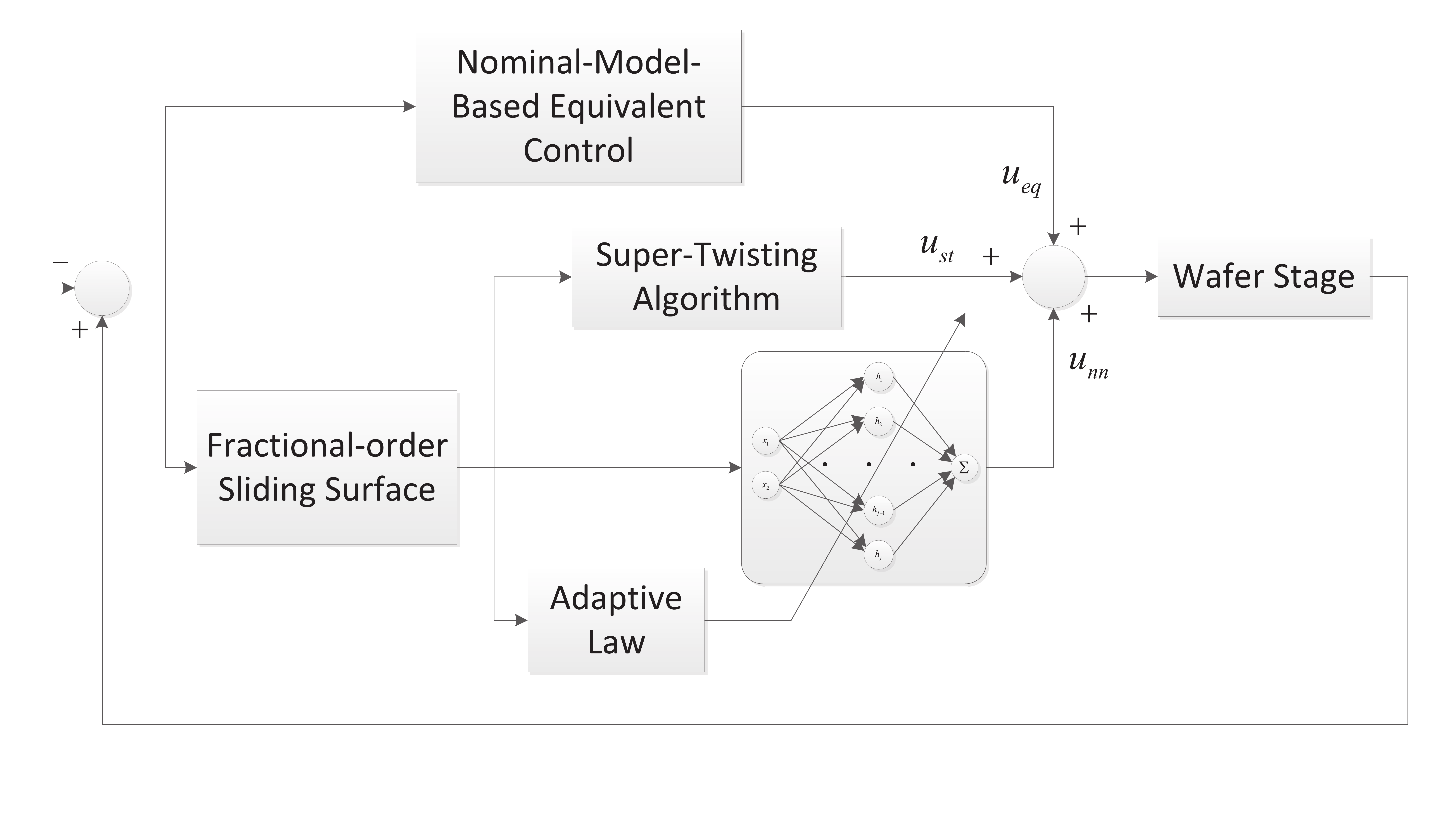}
	\caption{The proposed control structure of the system.}
	\label{fig:control structure}
\end{figure}

The fractional-order sliding surface is designed as
\begin{align} \label{definition of sliding surface}
    z=\dot{e}+\alpha_1 D^{\eta-1}(\textrm{sig}^a(e))+\alpha_2 e,
\end{align}
where $e=p-r$ is the tracking error, $r$ is the reference position, $\dot{e}$ is the derivative of $e$, $\alpha_1$ and $\alpha_2$ are two selected positive constants. $\textrm{sig}^a(e)\triangleq\textrm{sgn}(e)|e|^a$, and $\eta\in(0,1)$ is a pre-defined constant. Therefore, the control law for the equivalent controller can be obtained by setting $\dot{z}=0$, which yields:
\begin{align}
    u_{eq}=
    \frac{1}{\bar{B}}(\ddot{r}-\alpha_1D(D^{\eta-1}(\textrm{sig}^a(e)))-\alpha_2\dot{e}-\bar{A}v).
\end{align}

The control law of the super-twisting controller is given by
\begin{align}
        u_{st}=-\frac{h_1}{\bar{B}}\Phi_1(z)-\int_0^t\frac{h_2}{\bar{B}}\Phi_2(z) dt,
\end{align}
where $h_1$ and $h_2$ are tunable constants, and
\begin{align}
    \Phi_1(z)&=|z|^{\frac{1}{2}} \textrm{sgn}(z), \\
    \Phi_2(z)&=\frac{1}{2} \textrm{sgn}(z).
\end{align}

Applying the controller~(\ref{eq:overall controller}) to the wafer stage system~(\ref{model:final}), we obtain
\begin{align}
    \dot{z}= -h_1\Phi_1(z)-\int_0^t h_2 \Phi_2(z) dt+\bar{B}u_{nn}+f,
\end{align}
where $f=\frac{\delta_B}{1+\delta_B}(\ddot{r}-\alpha_1D(D^{\eta-1}(\textrm{sig}^a(e)))-\alpha_2\dot{e}-Av-d)-\frac{1}{\bar{B}}(\bar{A}\delta_A+d)$.

In order to reduce the influence of the uncertainties on the dynamics of sliding variable $z$, $u_{nn}$ is designed as 
\begin{align} \label{eq: initial nn controller}
    u_{nn}=-\frac{f}{\bar{B}}.
\end{align}

However, as unknown parameters exist in $f$, $u_{nn}$ cannot be directly obtained. Hence, in this paper, we propose to use an RBF neural network to approximate the value of $f$.

The RBF neural network contains three layers: an input layer, a hidden layer, and an output layer. The input layer is defined as 
\begin{align}
    \bm{x}=[x_1~x_2]^{\bm{T}}=[z~\dot{z}]^{\bm{T}}.
\end{align}
The hidden layer consists of Gaussian functions, and it is presented as
\begin{align}
\bm{h}(\bm{x})&=[h_1, h_2, \cdots h_{J-1}, h_{J}]^{\bm{T}}, \\
    h_j(\bm{x})=&\textrm{exp} \left(-{\frac{||\bm{x}-\bm{c_j}||^2}{2b_j^2}}\right),j=1,2,...,J.
\end{align}
where $\bm{c_j}=[c_{1j}, c_{2j}]^{\bm{T}}$ and $b_j$ are, respectively, the mean vector and the standard deviation of the $j$-th Gaussian basis. $||\bullet||$ denotes the Euclid norm of the vector $\bullet$, and $J$ is the number of Gaussian basis.

The output of the neural network is 
\begin{align}
    f=\bm{W}^{*T}\bm{h}(\bm{x})+{\varepsilon},
\end{align}
where $\bm{W}^{*}=(W_1^*,W_2^*,...,W_J^*)$ represents the optimal weight vector of the neural network, and $\varepsilon$ stands for the approximation error.


Traditional methods such as the gradient descent method to obtain the parameters of the RBF neural network cannot guarantee the global convergence of the closed-loop system~\citep{cuong2018adaptive}. To guarantee the effectiveness of the approximation and the global convergence of the closed-loop system, the mean vector $\bm{c_j}$ is designed to be in the effective mapping of the Gaussian function, $b_j$ is selected with a proper value~\citep{jinkun2013radial}, and the weight $\bm{W}$ is updated online by the adaptive law
\begin{align}\label{eq:adaptive law}
  \bm{ \dot{\hat{W}}}=(1+\rho)\Phi_2(z)\bm{h}(x)-\frac{d \Phi_1(z)}{dz} \omega,
\end{align}
where $\rho$ is a positive tunable parameter, $\bm{\hat{W}}$ is the approximation of $\bm{W}$, and $\omega=-h_2 \int_0^t \Phi_2(z)$.

Consequently, the approximation of $f$ is obtained as
\begin{align}
    \hat{f}=\bm{\hat{W}}^{\bm{T}} \bm{h}(x).
\end{align}

The neural-network-based controller~(\ref{eq: initial nn controller}) is rewritten as
\begin{align}
    u_{nn}=-\frac{\hat{f}}{\bar{B}}.
\end{align}

\subsection{Stability Analysis} \label{subsec: Stability Analysis}
\begin{mythm} \label{thm: stability of s}
For the system in (\ref{model:final}), if the controller is designed as~(\ref{eq:overall controller}), the parameters $h_1$ and $h_2$ are selected as
\begin{align} \label{ineq:h_1}
h_1 &\geq \max \left[ \rho^2+\rho+\frac{1}{2}, \frac{\rho^2+3 \rho +1}{2}\right],\\ \label{eq:h_2}
h_2  &= \rho^2+\rho(1+h_1),
\end{align}
and the adaptive law is designed as~(\ref{eq:adaptive law})
, then $\Phi_1(z)$ converges to the region
\begin{align}
    |\Phi_1(z)| \leq \frac{|\varepsilon|}{\min \left[\sqrt{2h_1-\rho}-1-\rho, \rho+\frac{2h_1-1}{2(\rho+1)} \right]}.
\end{align}
\end{mythm}

\begin{proof}
Substituting the designed controller~(\ref{eq:overall controller}) and the definition of the sliding surface~(\ref{definition of sliding surface}) into the model of the wafer stage~(\ref{model:final}), we have
\begin{align}\label{dynamics of s in proof}
\dot{z}=&-h_1\Phi_1(z)+\omega+\tilde{f},  \\
\dot{\omega}=&-h_2 \Phi_2(z).
\end{align}
where $\tilde{f}=f-\hat{f}$.

Select a Lyapunov function candidate as
  \begin{align}
 V =&\bm{\chi^T}\bm{P}\bm{ \chi} + \rho \bm{\tilde{W}}^{\bm{T}}\bm{\tilde{W}}, \label{eq:V}
\end{align}
  where $\bm{\chi}=
    \left[
    \begin{array}{cc}
      \Phi_1(z) & \omega
    \end{array}
    \right]^{\bm{T}}$, $\bm{\tilde{W}}=\bm{W}^*-\bm{\hat{W}}$, $\bm{P}$ is defined as
\begin{align}
   \bm{P}=
  \left[
    \begin{array}{cc}
      \rho +\rho ^2 & -\rho  \\
      -\rho  & 1
    \end{array}
  \right].
\end{align}

The derivative of $V$ is
 \begin{align}
\dot{V}=\left(\frac{d \bm{\chi}}{dt}\right)^{\bm T} \bm{P \chi}+\bm{\chi^T P}\frac{d \bm{\chi}}{dt}+2\rho \bm{\tilde{W}}^{\bm{T}}\dot{\bm{\tilde{W}}},
    \label{eq:derivative of V}
 \end{align}
in which the derivative of $\bm{\chi}$ is 
\begin{align}
\frac{d \bm{\chi}}{dt}=&
\left[
\begin{array}{cc}
  \frac{\Phi_1(z)}{dz}\dot{z} & \dot{\omega}
\end{array}
\right]^{\bm{T}} 
\nonumber \\
=&
\frac{\Phi_1(z)}{dz}
\left[
\begin{array}{cc}
-h_1 & 1\\
-h_2 & 0
\end{array}
\right]\left[
  \begin{array}{cc}
    \Phi_1(z) &
    \omega
  \end{array}
\right]^{\bm{T}}+
\left[
\begin{array}{cc}
   \frac{\Phi_1(z)}{dz}\tilde{f}  &  0
\end{array}
\right]^{\bm{T}}
\nonumber \\
=&
\frac{\Phi_1(z)}{dz}
\left[
\begin{array}{cc}
-h_1+\frac{\varepsilon}{\Phi_1(z)} & 1\\
-h_2 & 0
\end{array}
\right]\left[
  \begin{array}{cc}
    \Phi_1(z) &
    \omega
  \end{array}
\right]^{\bm{T}}
\nonumber \\
~&+
\left[
\begin{array}{cc}
   \frac{\Phi_1(z)}{dz}\bm{\tilde{W}}\bm{h}(x)  &  0
\end{array}
\right]^{\bm{T}}
. \label{eq:derative of chi}
\end{align} 
Substituting~(\ref{eq:derative of chi}) into~(\ref{eq:derivative of V}), we have
\begin{align} 
\dot{V}=&\frac{\Phi_1(z)}{dz} 
\bm{\chi}^{\bm{T}}
\left[
\begin{array}{cc}
     A & B \\
     B & -\rho 
\end{array}
\right]{\bm{\chi}}
-\frac{\Phi_1(z)}{dz}\bm{\chi}^{\bm{T}}
\left[
\begin{array}{cc}
     \rho&0  \\
     0& \rho
\end{array}
\right] 
\bm{\chi} \nonumber \\
~&+2\left((\rho+\rho^2)\Phi_2(z)\bm{\tilde{W}}^{\bm{T}}\bm{h}(x)\right.
\nonumber \\
~&\left.-\rho\frac{d \Phi_1(z)}{dz}\omega \bm{\tilde{W}}^{\bm{T}}\bm{h}(x)\right)
-2\rho\bm{\tilde{W}}^{\bm{T}}\dot{\bm{\hat{W}}} \nonumber \\
\label{eq:proof, dot V}
=& \frac{\Phi_1(z)}{dz} 
\bm{\chi}^{\bm{T}}
\left[
\begin{array}{cc}
     A & B \\
     B & -\rho 
\end{array}
\right]{\bm{\chi}}^{\bm{T}} 
-\frac{\Phi_1(z)}{dz}\bm{\chi}^{\bm{T}}
\left[
\begin{array}{cc}
     \rho&0  \\
     0& \rho
\end{array}
\right] 
\bm{\chi},
\end{align}
where $A=2(\rho+\rho^2)(-h_1+\frac{\varepsilon}{\Phi_1(z)})+2\rho h_2+\rho$, $B=\rho +\rho^2 -h_2 +h_1 \rho -\frac{\varepsilon}{\Phi_1(z)}\rho$. Further, taking~(\ref{ineq:h_1}) and~(\ref{eq:h_2}) into consideration, we have $A=2\rho (1+\rho)\frac{\varepsilon}{\Phi_1(z)}+2\rho^3+2\rho^2+\rho-2h_1\rho$, and $B=-\frac{\varepsilon}{\Phi_1(z)}\rho$.

When $|\Phi_1(z)|>\frac{|\varepsilon|}{\min \left[\sqrt{2h_1-\rho }-1-\rho, -\rho+\frac{2h_1-1}{2(\rho+1)} \right]}$, two inequations can be obtained
\begin{align} \label{ineq:1}
    \frac{|\varepsilon|}{|\Phi_1(z)|}&<-\rho+\frac{2h_1-1}{2(\rho+1)}, \\
    \frac{|\varepsilon|}{|\Phi_1(z)|}&<\sqrt{2h_1-\rho}-1-\rho.
\end{align}
Based on the definition of $A$ and (\ref{ineq:1}), we have $A < 0$.

For the condition that $1+\rho-\sqrt{2h_1-\rho}<\frac{\varepsilon}{\Phi_1(z)}<\sqrt{2h_1-\rho}-1-\rho$,
\begin{align}
    \left|
\begin{array}{cc}
    A & B \\
    B & -\rho
\end{array}
    \right|=&-\rho^2\frac{\varepsilon^2}{\Phi_1^2(z)}-2\rho^2(1+\rho)\frac{\varepsilon}{\Phi_1(z)}
    \nonumber\\
    ~&-2\rho^2(\rho^2+\rho+\frac{1}{2}-h_1) 
    > 0.
\end{align}

Then we can claim that the matrix $\left[\begin{array}{cc}
    A & B \\
    B & -\rho
\end{array}\right]$ is negative definite. Therefore,
from~(\ref{eq:proof, dot V}), we have
\begin{align}
\dot{V}
< &
-\frac{\Phi_1(z)}{dz}\bm{\chi}^{\bm{T}}
\left[
\begin{array}{cc}
     \rho&0  \\
     0& \rho
\end{array}
\right]
\bm{\chi}^{\bm{T}}
< 0.
\end{align}

Here completes the proof of Theorem~\ref{thm: stability of s}.
\end{proof}

\section{Experiments} \label{sec:experiments}
\subsection{Experimental Setup} \label{subsec:experimental setup}

\begin{figure}[http]
  \centering
   \includegraphics[width=250pt]{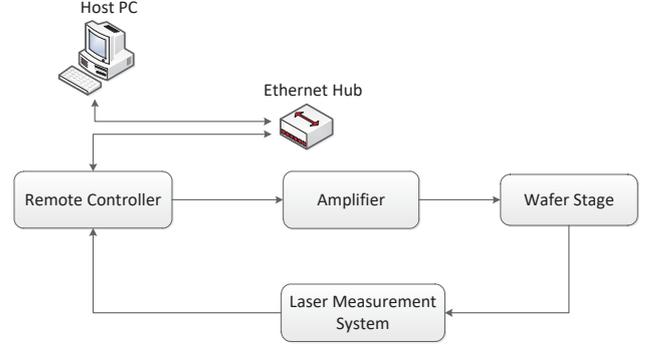}
   \caption{The block diagram of the experimental system.}
   \label{fig:structure_closed_loop}
   \end{figure}

The overall structure of our experimental system is depicted in Fig.~\ref{fig:structure_closed_loop}. The control algorithm is programmed in the LabView environment on the host computer. The host computer is connected with the remote controller (PXI 7831, from National Instruments) via Ethernet, so that the control algorithm can be deployed in the LabView Real-Time system in the remote controller. The output of the controller is amplified by an amplifier~(TA330, from Trust Automation) and applied to the wafer stage testbed. The position of the moving part of the wafer stage is measured by a laser ranging system~(from Keysight), and the measuring results are fed back to the remote controller. The nominal parameters of the wafer stage have been identified as $\bar{A}=-1.092~s^{-1}$ and $\bar{B}=3.9124~m/(s^2\cdot A)$.

\subsection{Experimental Results} \label{subsec:experimental results}
We implement the traditional PID controller, the SMC, the advanced FOSTA, and the proposed ANN-FSA to the wafer stage testbed to investigate the effectiveness of the proposed controller. The reference trajectory is shown in Fig.~\ref{fig:reference}. The scan length is set as $0.04$ m, and the scan velocity is set as $0.032$~m/s. The parameters in each controller are tuned so that the best performance of each controller is achieved. The sampling interval of the experiments are set as 1 ms. For the RBF neural network, the number of the hidden nodes is set as $5$, other parameters are selected as $c_1=[-3~-1~0~1~3]$, $c_2=[-7~-3~0~3~7]$, $b=[50~50~50~50~50]$, $\rho=0.2$ and the initial value of $\bm{W}$ is set as $\bm{0}$. In the fractional-order super-twisting algorithm, $\eta$ and $a$ are selected as $\frac{1}{2}$, and other parameters are tuned as $h_1=500$, $h_2=30$, $\alpha_1=0.001$ and $\alpha_2=175$.
   
\begin{figure}[http]
  \centering
  \includegraphics[width=250pt]{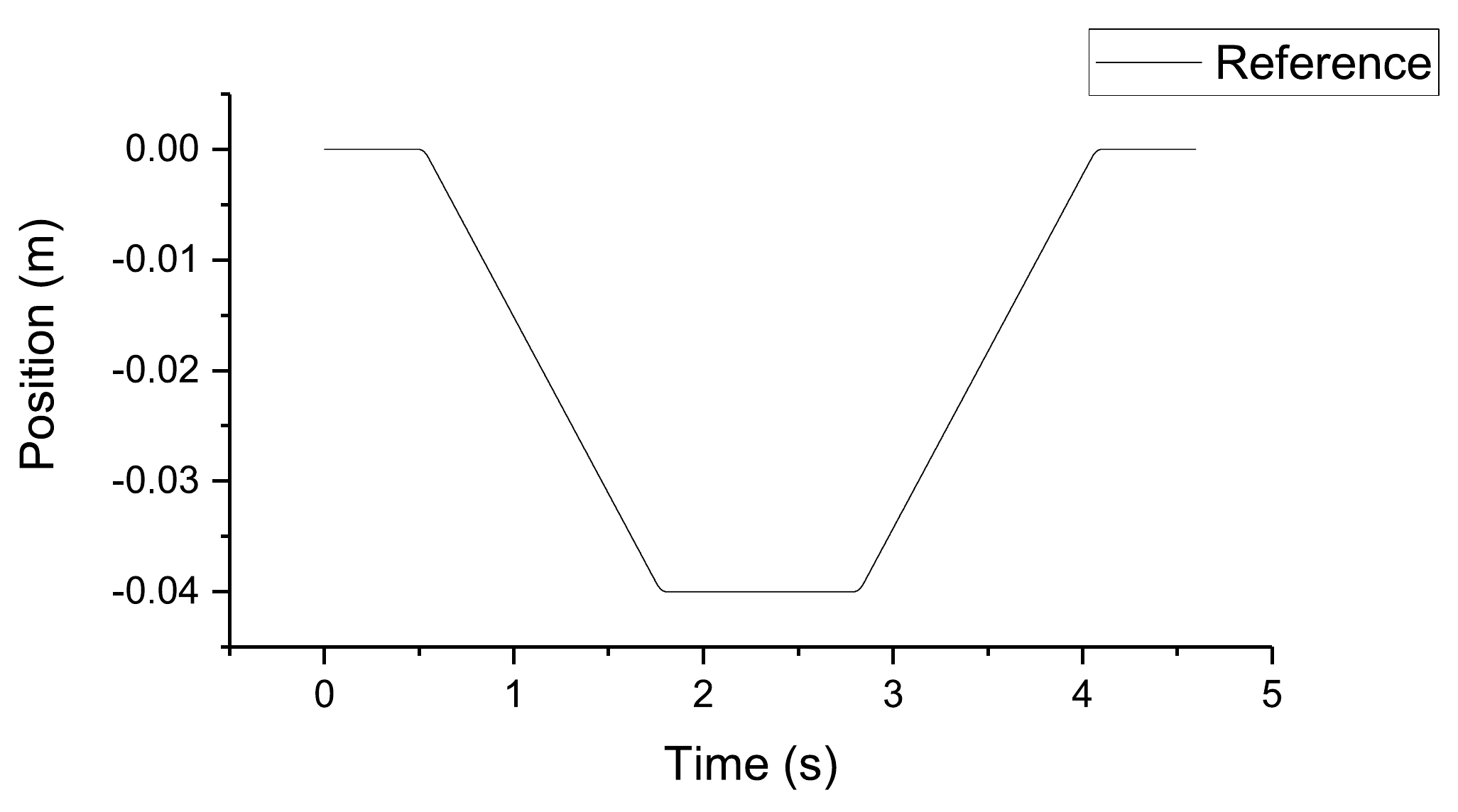}
  \caption{The reference signal in the experiments.}
  \label{fig:reference}
  \end{figure}
   
\begin{figure}[http]
  \centering
  \includegraphics[width=250pt]{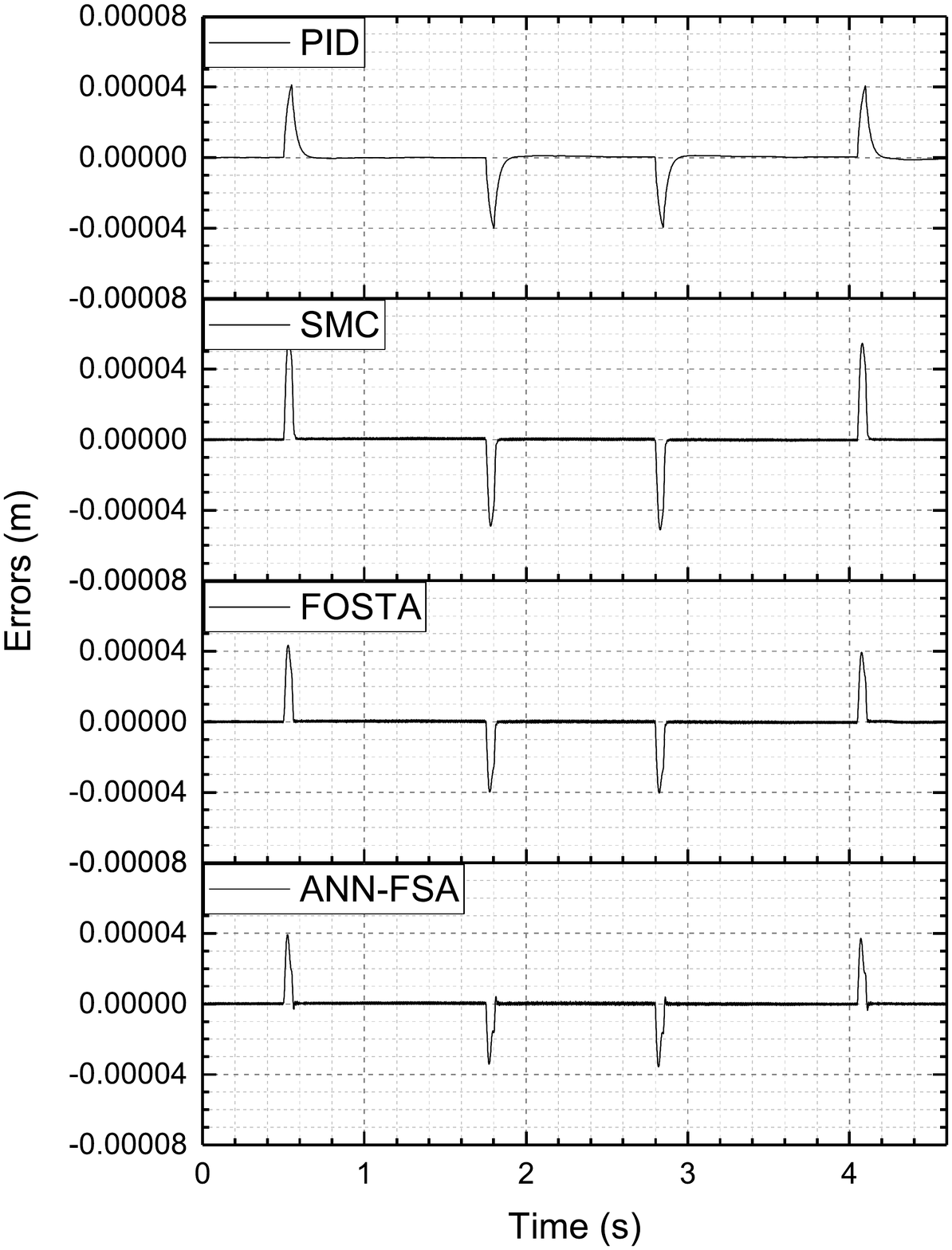}
  \caption{The tracking errors of the controllers in Case 1: without extra disturbance.}
  \label{fig:errors in case 1}
  \end{figure}

\begin{figure}[http]
  \centering
  \includegraphics[width=250pt]{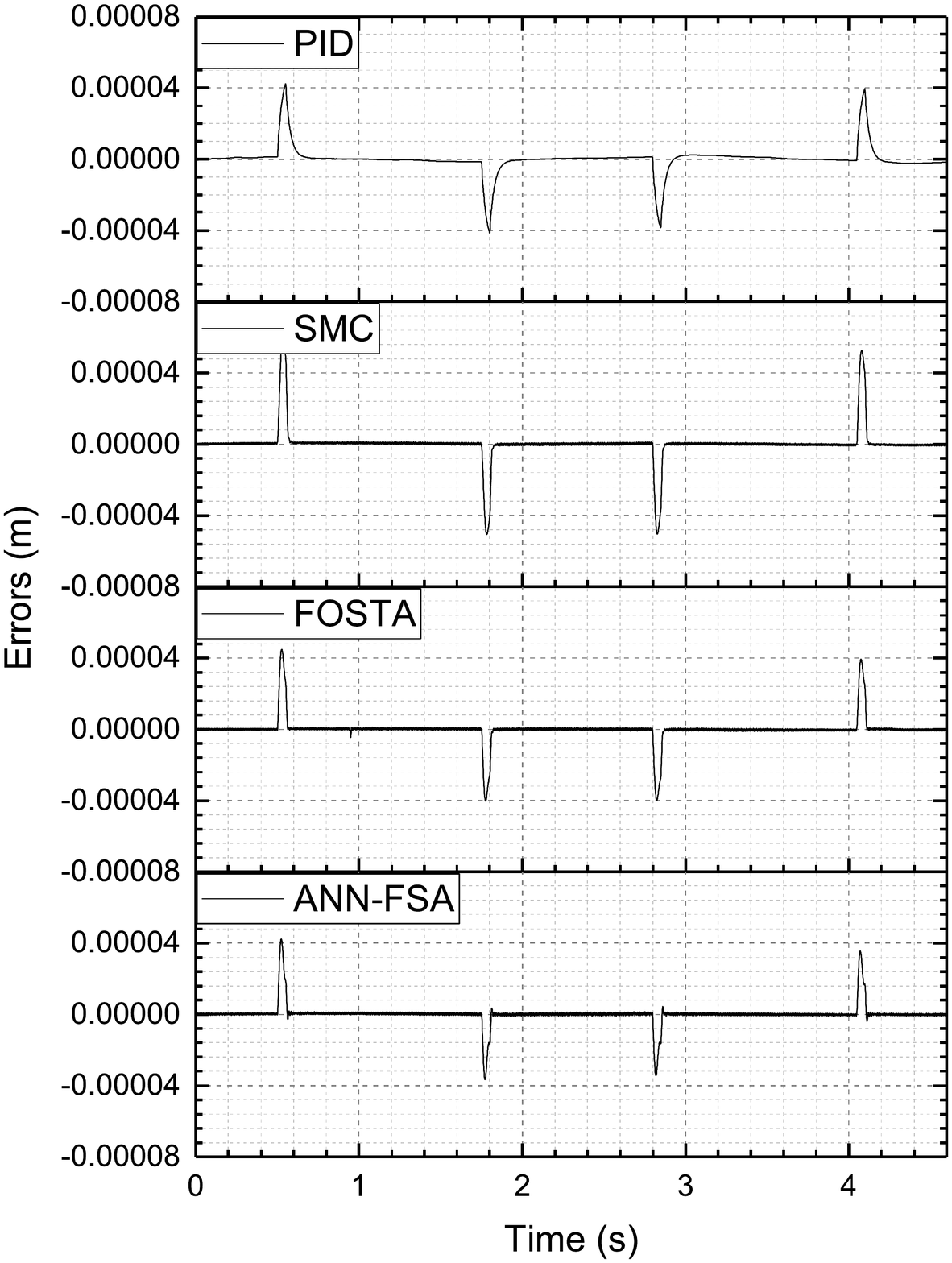}
  \caption{The tracking errors of the controllers in Case 2: with extra disturbance.}
  \label{fig:errors in case 2}
  \end{figure}
   
   Moreover, to study the robustness of these controllers, with all the parameters maintained the same, an extra external sinusoidal disturbance is generated and applied to the system. The amplitude and frequency of the disturbance signal are set as $0.03$~m (rather large compared with referce signal) and $1$~Hz, respectively. We denote the situation without extra disturbance as Case~1 and the situation with the additional disturbance as Case~2. The tracking performance in these two cases is shown in Fig.~\ref{fig:errors in case 1} and Fig.~\ref{fig:errors in case 2}, respectively.
   
   \begin{table}[tp]
   \centering
   \caption{RMS errors and the difference values between the two cases (Units: $\mu m$).}
   \label{table:RMS Error}
   \begin{tabular}{cccc}
   \cline{1-4}
    & Case 1 & Case 2 & Difference Values\\
    \cline{1-4}
    PID & 7.43 & 7.51 & 0.08 \\
    SMC & 9.12 & 9.15 & 0.03 \\
    FOSTA & 6.71 & 6.76 & 0.05 \\
    ANN-FSA & 4.95 & 4.94 & -0.01 \\
    \cline{1-4}
    \end{tabular}
    \end{table}
    \linespread{1}
   
   In Fig.~\ref{fig:errors in case 1}, we note that all the four controllers have large tracking errors when the scanning velocity changes. The peak error of SMC is the largest among the four controllers, which is at around $50$~$\mu$m. Tracking error via the proposed ANN-FSA is the smallest, about $35$~$\mu$m. We also note that the errors of SMC and ANN-FSA decay faster than the PID controller, but that the error via the PID controller is smoother than the other two controllers when the tracking errors are small, i.e., closer to zero. Figure \ref{fig:errors in case 2} shows that even in the presence of disturbances, the ANN-FSA can achieve the smallest tracking error. Moreover, from Fig.~\ref{fig:errors in case 1} and Fig.~\ref{fig:errors in case 2}, we note that the tracking error increases when disturbances exist. This means that the PID controller is not as robust enough as SMC and ANN-FSA. To quantitatively describe the robustness, we calculate the root mean square (RMS) errors of each controller in both cases and the results are displayed in Table.~\ref{table:RMS Error}. We can see that compared with traditional SMC, the precision of FOSTA is significantly improved. This is due to the introduction of the fractional-order sliding surface and the super-twisting algorithm. Further, the proposed ANN-FSA has the smallest RMS error in both Case~1 and Case~2. According to the difference values between the RMS errors in Case~1 and Case~2, we can conclude that SMC and FOSTA strategy is more robust than the PID controller. ANN-FSA remains precise in the presence of the large external disturbances. The results and comparison prove that the proposed control scheme achieves the best performance in terms of precision and robustness.

\section{Conclusion}  \label{sec:conclusion}
In this paper, an adaptive neural-network and fractional-order super-twisting algorithm was proposed and applied to a precision motion system. In this way, not only the dynamics of the states on the sliding surface was improved via the super-twisting algorithm, but also unknown model uncertainties and disturbances of the system were well compensated. 
Moreover, an adaptive law was derived for the neural-network-based controller so that the closed-loop system is globally convergent. Both stability analysis and experimental verification were provided. The comparison results among a PID controller, a conventional SMC, an advanced FOSTA and the proposed ANN-FSA showed that the proposed controller could achieve higher precision and better robustness than conventional controllers.



\bibliography{ifacconf}             
                                                   







\end{document}